\documentclass[11pt]{article}
\usepackage[utf8]{inputenc}
\usepackage[american]{babel}
\usepackage[textwidth=6.5in,textheight=9in]{geometry}
\usepackage{setspace}
\usepackage{amsmath,amssymb,amsthm}
\usepackage{natbib}
\usepackage[colorlinks=true,citecolor=blue]{hyperref}
\usepackage{graphicx}
\usepackage{subcaption}

\newtheorem{assumption}{Assumption}
\newtheorem{proposition}{Proposition}
\newtheorem{theorem}{Theorem}
\newtheorem{lemma}{Lemma}

\theoremstyle{remark}
\newtheorem{remark}{Remark}

\DeclareMathOperator{\expit}{expit}

\begin{document}

\title{Positivity-free Policy Learning with Observational Data}

\author{Pan Zhao$^1$\thanks{Email: \href{mailto:pan.zhao@inria.fr}{pan.zhao@inria.fr}}, Antoine Chambaz$^2$, Julie Josse$^1$, Shu Yang$^3$ \\
    $^1$ PreMeDICaL, Inria \\
    Desbrest Institute of Epidemiology and Public Health, University of Montpellier, France \\
    $^2$ Universit\'e Paris Cit\'e, CNRS, MAP5, F-75006 Paris, France \\
    $^3$ Department of Statistics, North Carolina State University, United States}

\maketitle

\begin{abstract}

Policy learning utilizing observational data is pivotal across various domains, with the objective of learning the optimal treatment assignment policy while adhering to specific constraints such as fairness, budget, and simplicity. This study introduces a novel positivity-free (stochastic) policy learning framework designed to address the challenges posed by the impracticality of the positivity assumption in real-world scenarios. This framework leverages incremental propensity score policies to adjust propensity score values instead of assigning fixed values to treatments. We characterize these incremental propensity score policies and establish identification conditions, employing semiparametric efficiency theory to propose efficient estimators capable of achieving rapid convergence rates, even when integrated with advanced machine learning algorithms. This paper provides a thorough exploration of the theoretical guarantees associated with policy learning and validates the proposed framework's finite-sample performance through comprehensive numerical experiments, ensuring the identification of causal effects from observational data is both robust and reliable.

\end{abstract}

\vfill
\newpage

\section{Introduction}

Over the past decade, methodologies for learning treatment assignment policies have seen substantial advancements in fields like biostatistics \citep{luedtke2016statistical,tsiatis2019dynamic}, computer science \citep{uehara2022review,yu2022estimating}, and econometrics \citep{athey2021policy,jia2023bayesian}. 
The core objective of data-driven policy learning is to learn optimal policies that map individual characteristics to treatment assignments to optimize some utility or outcome functions. This is crucial for deriving robust and trustworthy policies in high-stakes decision-making settings, requiring adherence to standard causal assumptions: consistency, unconfoundedness, and positivity \citep{van2011targeted,imbens2015causal}.

Various statistical and machine-learning methods have been developed to address policy learning tasks. Popular approaches include model-based methods such as Q-learning and A-learning \citep{murphy2003optimal,shi2018high}, and direct model-free policy search methods such as decision trees and outcome weighted learning \citep{zhang2012robust,cui2017tree}, among others \citep{bibaut2021risk,zhou2023offline}. Another prevailing line of work concerns heterogeneous treatment effects estimation \citep{wager2018estimation,kunzel2019metalearners,nie2021quasi,kallus2023robust}, where the sign of the conditional average treatment effects equivalently determines the optimal policy.

However, most methods depend heavily on the three standard causal assumptions to identify causal effects and optimal policies. Recent progress has been made to relax the consistency and unconfoundedness assumptions \citep{cortez2022staggered,kallus2018confounding}, but advancements addressing the violation of the positivity assumption are scarce. \cite{yang2018asymptotic} and \cite{branson2023causal} provide estimation and asymptotic inference results for propensity score trimming with binary and continuous treatments. \cite{lawrence2017counterfactual} consider counterfactual learning from deterministic bandit logs under lack of sufficient exploration. \cite{gui2023causal} use supervised representation learning to estimate causal effects for text data with apparent overlap violation. \cite{zhang2023semi} consider a missing-at-random mechanism without a positivity condition for generalizable and double robust inference for average treatment effects under selection bias with decaying overlap. \cite{jin2022policy} use pessimism and generalized empirical Bernstein's inequality to study offline policy learning without assuming any uniform overlap condition. \cite{khan2023off} provide partial identification results for off-policy evaluation under non-parametric Lipschitz smoothness assumptions on the conditional mean function, and thus avoid assuming either overlap or a well-specified model. \cite{liu2023average} propose the overlap weighted average treatment effect on the treated under lack of positivity. To our knowledge, our work is the first to consider learning treatment assignment policies while avoiding the positivity assumption.

This study introduces a novel positivity-free policy learning framework focusing on dynamic and stochastic policies, which are practical. We propose \emph{incremental propensity score policies} that shift propensity scores by an individualized parameter, requiring only the consistency and unconfoundedness causal assumptions. Our approach enhances the concept of incremental intervention effects, as proposed by \cite{kennedy2019nonparametric}, adapting it to individual treatment policy contexts.

We also use semiparametric theory to characterize the efficient influence function \citep{bickel1993efficient,laan2003unified}, which serves as the foundation to construct estimators with favorable properties, such as double/multiple robustness and asymptotically negligible second-order bias (also called Neyman orthogonality in double machine learning \citep{chernozhukov2018double} or orthogonal statistical learning \citep{foster2023orthogonal}). Thus, our proposed estimators can attain fast parametric $\sqrt{n}$ convergence rates, even when nuisance parameters are estimated at slower rates such as $n^{1/4}$ via flexible machine learning algorithms.

Based on the above efficient off-policy evaluation results, we propose approaches to learning the optimal policy by maximizing the value function, possibly under application-specific constraints. Several examples are provided in Section~\ref{sec:meth}, including fairness and resource limit. While it remains an open problem to provide finite sample or asymptotic regret bounds as \cite{athey2021policy} for stochastic policy learning with constraints, which is out of the scope of this article, we establish asymptotic guarantees for our proposed policy learning methods under alternative (stronger) conditions.

The rest of this article is organized as follows. Section~\ref{sec:sf} introduces the basic setup and notations and proposes the incremental propensity score policy. Our main identification and semiparametric efficiency theory results for off-policy evaluation are presented in Section~\ref{sec:ope.eff}. Section~\ref{sec:meth} formally introduces our positivy-free policy learning framework, with several examples. Asymptotic analysis of guarantees for policy evaluation and learning are given in Section~\ref{sec:asym}. Finally, we illustrate our methods via simulations and a data application in Section~\ref{sec:simu}. The article concludes in Section~\ref{sec:disc} with a discussion of some remarks and future work. All proofs and additional results are provided in the Supplementary Material.

\section{Statistical Framework}
\label{sec:sf}

We first introduce the notations and setup. Let $X$ denote the $p$-dimensional vector of covariates that belongs to a covariate space $\mathcal{X} \subset \mathbb{R}^p$, $A \in \mathcal{A} = \{0, 1\}$ denote the binary treatment, $Y \in \mathbb{R}$ denote the outcome of interest. Without loss of generality, we assume throughout that larger values of $Y$ are more desirable. Our observed data structure is $O = (X, A, Y)$. Suppose that our collected random sample $(O_{1}, \ldots, O_{n})$ of size $n$ are independent and identically distributed (i.i.d.) observations of $O \sim P$, where $P$ denote the true distribution of the observed data.

Now, we are in the position to introduce different types of policies or interventions commonly used in the literature: (i) under \emph{static} policies, the same treatments would be applied indiscriminately, while \emph{dynamic} policies depend on individual characteristics; (ii) \emph{deterministic} policies recommend one specific treatment and \emph{stochastic} policies output probabilities of prescribing each treatment level. This article focuses on dynamic and stochastic policies, which are more practical in various settings and have received substantial recent interest. Typical examples include point exposures \citep{dudik2014doubly}, longitudinal studies \citep{tian2008identifying,murphy2001marginal,van2007causal}, natural stochastic policies in reinforcement learning \citep{kallus2020efficient}, and particularly interventions that depend on the observational treatment process \citep{munoz2012population,haneuse2013estimation,young2014identification}; but none of the existing intervention effects both avoids positivity conditions entirely and is completely nonparametric.

We use the potential outcomes framework \citep{neyman1923applications,rubin1974estimating} to define causal effects. Let $Y(a)$ denote the potential outcome had the treatment $a$ been assigned. A policy $d: \mathcal{X} \to \{0, 1\}$ is deterministic if it maps individual characteristics $x$ to a treatment assignment $0$ or $1$, and the output of a stochastic policy $d: \mathcal{X} \to [0, 1]$ is the probability of assigning treatment $1$. Let $\mathcal{D}$ denote a pre-specified class of policies of interest, where each policy $d \in \mathcal{D}$ induces the value function defined by
\begin{equation*}
V(d) = E[Y(d)] = E[Y(1) d(X) + Y(0) (1 - d(X))],
\end{equation*}
where $Y(d)$ is the potential outcome under the policy $d$. In Remark~\ref{rm:standard.policy.learning}, we briefly review standard (deterministic) policy learning methods. In our framework, we focus on dynamic and stochastic policies. Our goal is to directly search for the optimal policy $d^\ast$ that maximizes the value function $V(d)$, possibly under application-specific constraints $c(d) \leq 0$. See Section~\ref{sec:meth} for detailed examples.

\subsection{Causal Assumptions}

We make the following identification assumptions.

\begin{assumption}[Consistency]\label{asmp:cons}
$Y = Y(A)$.
\end{assumption}

\begin{assumption}[Unconfoundedness]\label{asmp:unconf}
$A \perp Y(a) \mid X$ for $a = 0, 1$.
\end{assumption}

Assumption~\ref{asmp:cons} is also known as the stable unit treatment value assumption, which says there should be no multiple versions of the treatment and no interference between units. Assumption~\ref{asmp:unconf} states that there are no unmeasured confounders so that treatment assignment is as good as random conditional on the covariates $X$. In this article, we entirely avoid the positivity assumption which requires that each unit has a positive probability of receiving both treatment levels, i.e., $c < Pr(A = 1 \mid X) < 1 - c$ for some constant $c > 0$.

\begin{remark}\label{rm:standard.policy.learning}
Standard policy learning methods need all of Assumptions~\ref{asmp:cons}, \ref{asmp:unconf} and the positivity assumption to identify the value function of deterministic policies $d: \mathcal{X} \to \mathcal{A}$ by the outcome regression (OR), inverse probability weighting (IPW) and augmented IPW (AIPW) formulas:
\begin{align*}
V_{\rm OR}(d) &= E[E[Y \mid X, A = d(X)]], \quad
V_{\rm IPW}(d) = E\left[\frac{I\{A = d(X)\} Y}{Pr(A = d(X) \mid X)}\right], \\
V_{\rm AIPW}(d) &= E\left[E[Y \mid X, A = d(X)] + \frac{I\{A = d(X)\} (Y - E[Y \mid X, A = d(X)])}{Pr(A = d(X) \mid X)}\right],
\end{align*}
thus the optimal policies are given by $d_{\rm OR}^\ast = \arg\max_{d \in \mathcal{D}} V_{\rm OR}(d)$, $d_{\rm IPW}^\ast = \arg\max_{d \in \mathcal{D}} V_{\rm IPW}(d)$, and $d_{\rm AIPW}^\ast = \arg\max_{d \in \mathcal{D}} V_{\rm AIPW}(d)$, possibly under application-specific constraints. When the positivity is violated, it is error-prone to rely on the outcome regression model's extrapolation, and the IPW and AIPW estimators would fail due to division by zero.
\end{remark}

\subsection{Incremental Propensity Score Policies}

\cite{kennedy2019nonparametric} propose a new class of stochastic dynamic intervention, called incremental propensity score interventions, and show that these interventions are nonparametrically identified without requiring any positivity restrictions on the propensity scores. Specifically, their proposed intervention replaces the observational propensity score $\pi$ with a shifted version based on multiplying the odds of receiving treatment, ${\delta \pi(x)}/\{\delta \pi(x) + 1 - \pi(x)\}$, where the increment parameter $\delta \in (0, \infty)$ is user-specified and dictates the extent to which the propensity scores fluctuate from their actual observational values. Some motivation and examples, efficiency theory, and estimators for mean outcomes under these interventions are studied in detail by \cite{kennedy2019nonparametric}.

We propose a positivity-free (stochastic) policy learning framework based on the incremental propensity score interventions. Specifically, we consider the stochastic policy $d: \mathcal{X} \to [0, 1]$ that assigns treatment $1$ with probability 
\begin{equation}\label{eq:stoc.policy}
d(x) = \frac{\delta(x) \pi(x)}{\delta(x) \pi(x) + 1 - \pi(x)},
\end{equation}
where $\delta (x)$ enables individualized treatment assignment. We note that the choice of $d(x)$ in \eqref{eq:stoc.policy} is motivated by its interpretability and positivity-free. In particular, whenever $0 < \pi(x) < 1$, $\delta(x) = [{d(x) / \{1 - d(x)\}}]/[\pi(x) / \{1 - \pi(x)\}]$ is simply an odds ratio, indicating how the policy changes the odds of receiving treatment. When positivity is violated, we have that $d(x) = 0$ if $\pi(x) = 0$, and $d(x) = 1$ if $\pi(x) = 1$.

\section{Identification and Efficiency Theory}
\label{sec:ope.eff}

\subsection{Identification}

We first give formal identification results for the value function of incremental propensity score policies, which require no conditions on the propensity scores.

\begin{proposition}[Identification formulas]\label{prop:iden}
Under Assumptions~\ref{asmp:cons} and \ref{asmp:unconf}, the value function $V(d)$ can be nonparametrically identified by the outcome regression with incremental propensity score (OR-IPS) formula:
\begin{equation}\label{eq:value.or}
V_{\rm OR-IPS}(d) = E\left[\frac{\delta(X) \pi(X) \mu_1 (X) + \{1 - \pi(X)\} \mu_0(X)}{\delta(X) \pi(X) + 1 - \pi(X)}\right],
\end{equation}
where $\mu_{a} (X) = E[Y \mid X, A = a], a = 0, 1$ are the outcome regression functions or the inverse probability weighting of incremental propensity score (IPW-IPS) formula:
\begin{equation}\label{eq:value.ipw}
V_{\rm IPW-IPS}(d) = E\left[\frac{Y\{\delta(X) A + 1 - A\}}{\delta(X) \pi(X) + 1 - \pi(X)}\right].
\end{equation}
\end{proposition}

Proposition~\ref{prop:iden} shows that the value function can be identified by (i) a weighted average of the outcome regression functions $\mu_{0}, \mu_{1}$, where the weight on $\mu_{1}$ is given by the incremental propensity score $d(x)$ and the weight on $\mu_{0}$ is $1 - d(x)$; (ii) inverse probability weighting where each treated is weighted by the (inverse of the) propensity score plus some fractional contribution of its complement, i.e., $\pi(x) + (1 - \pi(x)) / \delta (x)$, and untreated units are weighted by this same amount, except the entire weight is further down-weighted by a factor of $\delta (x)$.

\subsection{Efficient Off-policy Evaluation}

Despite that simple plug-in OR-IPS and IPW-IPS estimators can be easily constructed from \eqref{eq:value.or} and \eqref{eq:value.ipw}, these estimators will only be $\sqrt{n}$-consistent when the outcome regression or propensity score models are correctly specified. This is usually unrealistic in practice. We use semiparametric efficiency theory to study the following statistical functional of $P$ from a nonparametric statistical model $\mathcal{M}$:
\begin{equation*}
\Psi (P) = V(d) = E_{P} \left[\frac{\delta(X) \pi(X) \mu_1 (X) + \{1 - \pi(X)\} \mu_0(X)}{\delta(X) \pi(X) + 1 - \pi(X)}\right],
\end{equation*}
and propose efficient estimators based on the efficient influence function.

\begin{proposition}[Semiparametric Efficiency]\label{prop:eif}
The efficient influence function of $\Psi (P)$ is
\begin{equation}\label{eq:value.eif}
\begin{split}
\phi (P) (O) &= \frac{A \delta(X) \{Y - \mu_1(X)\} + (1 - A) \{Y - \mu_0(X)\} + \delta(X) \pi(X) \mu_1 (X) + \{1 - \pi(X)\} \mu_0(X)}{\delta(X) \pi(X) + 1 - \pi(X)} \\ 
&\quad + \frac{\delta(X) \tau(X) \{A - \pi(X)\}}{\{\delta(X) \pi(X) + 1 - \pi(X)\}^2} - \Psi (P),
\end{split}
\end{equation}
where $\tau (x) = \mu_1(x) - \mu_0(x)$.
\end{proposition}

By Proposition~\ref{prop:eif}, the one-step bias-corrected estimator is given by
\begin{equation}\label{eq:one.step.est}
\hat{\Psi}_{\rm OS} = \Psi (\hat{P}) + P_{n} \phi (\hat{P}) (O) = \frac{1}{n} \sum_{i=1}^{n} \xi(\hat{P})(O_i),
\end{equation}
where we estimate $P$ by $\hat{P}$, and let $P_{n}$ denote the empirical distribution, and $\xi(P)(O) = \phi (P) (O) + \Psi (P)$ is the uncentered efficient influence function. This estimator can converge at fast parametric $\sqrt{n}$ rates and attain the efficiency bound, even when the propensity score $\pi (x)$ and outcome regression functions $\mu_{0}, \mu_{1}$ are modeled flexibly and estimated at rates slower than $\sqrt{n}$, as long as these nuisance functions are estimated consistently at rates faster than $n^{1/4}$. This allows much more flexible nonparametric methods and modern machine learning algorithms to be employed.

However, characterizing asymptotic properties of the estimator \eqref{eq:one.step.est} requires some empirical process conditions that restrict the flexibility and complexity of the nuisance estimators; otherwise, we will have overfitting bias and intractable asymptotic behaviors. See the asymptotic analysis in Section~\ref{sec:asym} and proofs thereof. To accommodate the wide use of modern machine learning algorithms that usually fail to satisfy the required empirical process conditions, we apply the cross-fitting procedure to obtain asymptotically normal and efficient estimators \citep{zheng2010asymptotic,chernozhukov2018double}. Suppose we randomly split the data into $K$ folds. Then the cross-fitting estimator is 
\begin{equation}\label{eq:cross.fit.est}
\hat{\Psi}_{\rm CF} = \frac{1}{K} \sum_{k=1}^{K} \hat{\Psi}_{k} = \frac{1}{K} \sum_{k=1}^{K} P_{n,k} \xi(P_{n,-k})(O),
\end{equation}
where $P_{n,k}$ and $P_{n,-k}$ denote the empirical measures on data from the $k$-fold and excluding the $k$-fold, respectively. That is, for $k = 1, \ldots, K$, nuisance estimators are constructed excluding the $k$-fold, and the value function $\hat{\Psi}_{k}$ is evaluated on the $k$-th fold; finally, the cross-fitting estimator is the average of the $K$ value estimators from $K$ folds.

\section{From Efficient Policy Evaluation to Learning}
\label{sec:meth}

In this section, we first present our proposed methods for policy learning.

As discussed in Section~\ref{sec:sf}, given a pre-specified policy class $\mathcal{D}$ (e.g., linear decision rules), we propose estimating the optimal treatment assignment rule $\hat{d}$ that solves (i) $\hat{d} = \arg\max_{d \in \mathcal{D}} \hat{V}(d)$, where $\hat{V}(d)$ is a value function estimator by OR-IPS \eqref{eq:value.or}, IPW-IPS \eqref{eq:value.ipw}, one-step \eqref{eq:one.step.est} or cross-fitting \eqref{eq:cross.fit.est}; or (ii) $\hat{d} = \arg\max_{d \in \mathcal{D}} \hat{V}(d)$ subject to $\hat{c}(d) \leq c$, when an application-specific constraint $c(d) \leq c$ is imposed, and $\hat{c}(d)$ is a constraint estimator which usually needs to be studied on a case-by-case basis.

We first review important examples of policy learning that fit into our framework.

\noindent{\bf Vanilla direct policy search.} The first example is what most existing work on policy learning has focused on, primarily for deterministic policies with a binary treatment. When the policy class is unrestricted, the optimal treatment assignment rule depends on the sign of the conditional average treatment effect for each individual unit, which cannot be extended to stochastic policies. Our proposed optimal incremental propensity score policies maximize the value function.

\noindent{\bf Fair policy learning.} In many decision-making scenarios, such as hiring, recommendation systems, and criminal justice, concerns have been raised regarding the fairness of decisions from the learning process \citep{chzhen2020fair}. Let $S \in \mathcal{S}$ denote the sensitive attribute. For randomized predictions $f: \mathcal{X} \times \mathcal{S} \to \Delta (\mathcal{A})$, popular fairness criteria include demographic parity (DP) \citep{calders2009building}:
\begin{equation}
E[f(X, S) \mid S = s] = E[f(X, S) \mid S = s'], \forall s, s' \in \mathcal{S},
\end{equation}
which says that $f(X, S)$ is independent from $S$, or equal opportunity (EO) \citep{hardt2016equality}:
\begin{equation}
E[f(X, S) \mid S = s, A = a] = E[f(X, S) \mid S = s', A = a], \forall s, s' \in \mathcal{S}, a \in \mathcal{A},
\end{equation}
which requires equal true positive and true negative rates. Following the same spirit, we consider fair policy learning tasks as the constrained optimization problem:
\begin{equation*}
\max_{d \in \mathcal{D}} V(d), \text{ subject to } f(d) \leq b,
\end{equation*}
where $f(d)$ is either the DP or EO metrics, which can be estimated by
\begin{equation*}
\hat{f}_{\rm DP}(d) = \left(\sum_{s \in \mathcal{S}} \left(\frac{\sum_{i=1}^{n} d(X_i) I\{S_i = s\}}{\sum_{i=1}^{n} I\{S_i = s\}} - \frac{\sum_{i=1}^{n} d(X_i)}{n} \right)^{2}\right)^{1/2},
\end{equation*}
or
\begin{equation*}
\hat{f}_{\rm EO}(d) = \left(\sum_{s \in \mathcal{S}} \left(\frac{\sum_{i=1}^{n} d(X_i) I\{S_i = s, A_i = 1\}}{\sum_{i=1}^{n} I\{S_i = s, A_i = 1\}} - \frac{\sum_{i=1}^{n} d(X_i) I\{A_i = 1\}}{\sum_{i=1}^{n} I\{A_i = 1\}} \right)^{2}\right)^{1/2},
\end{equation*}
and $b$ is a pre-specified tuning parameter.

\noindent{\bf Resource-limited policy learning.} In many real-world applications, the proportion of individuals who can receive the treatment is a priori limited due to a budget or a capacity constraint. So we consider the resource-limited policy learning tasks as the constrained optimization problem:
\begin{equation*}
\max_{d \in \mathcal{D}} V(d), \text{ subject to } E[d] \leq b,
\end{equation*}
where $b$ is the pre-specified budget or capacity.

\noindent{\bf Protect the vulnerable.} Since the optimal policy is typically defined as the maximizer of the expected potential outcome over the entire population, such a policy may be suboptimal or even detrimental to certain disadvantaged subgroups. \cite{fang2022fairness} propose the fairness-oriented optimal policy learning framework:
\begin{equation*}
\max_{d \in \mathcal{D}} V(d), \text{ subject to } Q_{\tau}(Y(d)) \geq b,
\end{equation*}
where $Q_{\tau}(Y(d)) = \inf\{t : F_{Y(d)}(t) \geq \tau\}$ is the $\tau$-th quantile of $Y(d)$, $F_{Y(d)}$ denotes the cumulative distribution function of $Y(d)$, and $b$ is a pre-specified protection threshold. Note that the quantile function can be estimated by $\hat{Q}_{\tau}(Y(d)) = \arg\min_{q} n^{-1} \sum_{i=1}^{n} c_{i}(d) \rho_{\tau} (Y_{i} - q)$, where $\rho_{\tau}(u) = u (\tau - I\{u < 0\})$ is the quantile loss function, and $c_{i}(d) = A_{i} d(X_{i}) + (1 - A_{i}) (1 - d(X_{i}))$.

Other examples in the literature include the counterfactual no-harm criterion by the principal stratification method \citep{li2023trustworthy}, (weakly) NP-hard knapsack problem \citep{luedtke2016optimal}, and instrumental variable methods \citep{qiu2021optimal}.

\section{Asymptotic Analysis of Policy Evaluation and Learning}
\label{sec:asym}

In this section, we first characterize the asymptotic distributions of our proposed one-step estimator~\eqref{eq:one.step.est} and the cross-fitted estimator~\eqref{eq:cross.fit.est} for off-policy evaluation.

\begin{theorem}\label{thm:donsker}
Assume the following conditions hold: {\rm (i)} $\|\hat{\pi}(x) - \pi(x)\|_{L_2} = o_p(n^{-1/4})$, $\|\hat{\mu}_{a} - \mu_{a}\|_{L_2} = o_p(n^{-1/4})$ for $a = 0, 1$; {\rm (ii)} $\phi (P)$ belongs to a Donsker class; {\rm (iii)} $|Y|$ and $|\delta (X)|$ are bounded in probability. For the one-step estimator, we have that $\sqrt{n} (\hat{\Psi}_{\rm OS} - \Psi (P)) \to \mathcal{N}(0, E[\phi^{2}])$.
\end{theorem}

\begin{theorem}\label{thm:cross.fit}
Assume the following conditions hold: {\rm (i)} $\|\hat{\pi}(x) - \pi(x)\|_{L_2} = o_p(n^{-1/4})$, $\|\hat{\mu}_{a} - \mu_{a}\|_{L_2} = o_p(n^{-1/4})$ for $a = 0, 1$; {\rm (ii)} $|Y|$ and $|\delta (X)|$ are bounded in probability. For the cross-fitting estimator, we have that $\sqrt{n} (\hat{\Psi}_{\rm CF} - \Psi (P)) \to \mathcal{N}(0, E[\phi^{2}])$.
\end{theorem}

Condition (i) of Theorems~\ref{thm:donsker} and \ref{thm:cross.fit} is commonly assumed such that the second-order remainder term is $o_p(1)$ \citep{kennedy2022semiparametric}. Condition (ii) of Theorems~\ref{thm:donsker} ensures the centered empirical process term is $o_p(1)$. Condition (iii) of Theorems~\ref{thm:donsker} and condition (ii) of Theorems~\ref{thm:cross.fit} are mild regularity conditions. The asymptotic variance of the one-step estimator can be consistently estimated by $\frac{1}{n} \sum_{i=1}^{n} \phi^{2} (\hat{P}) (O_{i})$, and the asymptotic variance of the cross-fitting estimator can be consistently estimated by $\frac{1}{K} \sum_{k=1}^{K} P_{n,k} \phi^{2} (\hat{P}_{-k}) (O)$.

Next, we prove asymptotic guarantees for the following generic off-policy learning problem:
\begin{equation*}
\max_{d \in \mathcal{D}} \hat{V}(d), \text{ subject to } \hat{c}(d) \leq c,
\end{equation*}
where $\hat{V}(d)$ is a value estimator of our proposed incremental propensity score policies, $\hat{c}(d)$ is an estimate of the constraint, and $c$ is a pre-specified criterion.

Consider a parametric policy class $\mathcal{D}(H)$ indexed by $\eta \in H$, where $H$ is a compact set. Let $\eta^\ast$ denote the true Euclidean parameter indexing the optimal policy. To simplify the notation, for $d(x;\eta) \in \mathcal{D}(H)$, we define $V(\eta) = V(d(x;\eta))$ and $c(\eta) = c(d(x;\eta))$.
\begin{theorem}\label{thm:cvx.policy}
Assume the following conditions hold: {\rm (i)} $d(x;\eta)$ is a continuously differentiable and convex function with respect to $\eta$; {\rm (ii)} $\hat{V}(\eta)$ and $\hat{c}(\eta)$ converge to $V(\eta)$ and $c(\eta)$ at rates $\sqrt{n}$. We have that {\rm (i)} $V(\hat{\eta}) - V(\eta^\ast) = O_p (n^{-1/2})$; {\rm (ii)} $\hat{V}(\hat{\eta}) - V(\eta^\ast) = O_p (n^{-1/2})$.
\end{theorem}

\begin{theorem}\label{thm:gc.policy}
Assume the following conditions hold: {\rm (i)} $\mathcal{D}$ is a Glivenko–Cantelli class; {\rm (ii)} $\hat{\pi}(x)$ and $\hat{\mu}_{a}(x)$ are uniformly consistent estimators of $\pi(x)$ and $\mu_{a}(x)$ for $a = 0, 1$; {\rm (iii)} $\forall d \in \mathcal{D}$, $m \in (0,1)$, it follows that $m d \in \mathcal{D}$. We have that {\rm (i)} $V(\hat{d}) - V(d) = o_p(1)$; {\rm (ii)} $\hat{V}(\hat{d}) - V(d) = o_p(1)$.
\end{theorem}

Theorem~\ref{thm:cvx.policy} (i) establishes that the regret of the learned policy attains the convergence rate of $n^{-1/2}$, and (ii) shows that $\hat{V}(\hat{\eta})$ is a $\sqrt{n}$-consistent estimator of the optimal value function for parametric and convex policy classes under mild assumptions. Theorem~\ref{thm:gc.policy} (i) establishes that the regret of the learned policy vanishes, and (ii) shows $\hat{V}(\hat{\eta})$ is still a consistent estimator for GC classes.

\section{Experiments}
\label{sec:simu}

In this section, we conduct extensive experiments to evaluate the performance of our proposed positivity-free policy learning methods by comparison with standard policy learning methods. Replication code is available at \href{https://github.com/panzhaooo/positivity-free-policy-learning}{GitHub}.

\subsection{Simulation}

We consider the fair policy learning task under the demographic parity constraint and simulate
\begin{align*}
    S & \sim \text{Bernoulli}(0.5), \quad
    (X_1, X_2, X_3)  \sim \text{Uniform}(0, 1), \\
    A & \sim \text{Bernoulli}(\text{expit}(-1 - X_1 + 1.5 X_2 - 0.25 X_3 - 3.1 S)), \\
    Y(0) & \sim \mathcal{N}\{20 (1 + X_1 - X_2 + X_3^2 + \exp{(X_2)}), 20^2\}, \\
    Y(1) & \sim \mathcal{N}\{20 (1 + X_1 - X_2 + X_3^2 + \exp{(X_2)}) + 25 (3 - 5 X_1 + 2 X_2 - 3 X_3 + S), 20^2\},
\end{align*}
where $\expit: x \mapsto 1 / (1 + \exp{(-x)})$. We let $S$ denote the sensitive attribute and $X_1, X_2, X_3$ the common non-sensitive attributes. The treatment assignment mechanism yields variable propensity scores that can degrade the performance of weighting-based estimators in standard policy learning methods. For standard methods, we consider the policy class of linear rules $\mathcal{D}_{\rm linear} = \{d(s, x) = I\{(1, s, x_1, x_2, x_3) \beta > 0\}: \beta \in \mathbb{R}^{5}, \|\beta\|_{2} = 1\}$. For the incremental propensity score policies, we consider the class $\mathcal{D}_{\rm IPS} = \{d(s, x) = \delta(s, x;\beta) \pi(s, x) / \{\delta(s, x;\beta) \pi(s, x) + 1 - \pi(s, x)\}: \beta \in \mathbb{R}^{5}\}$, which is indexed by $\delta(s, x;\beta) = \exp{\{(1, s, x_1, x_2, x_3) \beta\}}$.

We estimate the outcome regression model $\mu(s,x)$ and the propensity score $\pi(s,x)$ using the generalized random forests \citep{athey2019generalized} implemented in the \texttt{R} package \texttt{grf}. The constrained optimization problems are solved by the derivative-free linear approximations algorithm \citep{powell1994direct}, implemented in the \texttt{R} package \texttt{nloptr}. The sample size is $n = 1000$, and the demographic parity threshold is $\tau = 0.01$. 

We compare the true values of the estimated optimal policies using test data with sample size $N = 10^5$. The true optimal value is approximated using the test data. Simulation results of $100$ Monte Carlo repetition are reported in Figure~\ref{fig:fair-dp}. When some estimated propensity scores are exactly $0$, the IPW and AIPW estimators would fail, and \texttt{NA} is returned. Three standard methods IPW, OR, and AIPW have the worst performance. The IPW-IPS estimator also has large variability, which is similarly reported in \cite{kennedy2019nonparametric}. The OR-IPS and efficient one-step estimators achieve the best performance with the highest value.

\begin{figure}[ht]
     \centering
     \begin{subfigure}[b]{0.475\textwidth}
         \centering
         \includegraphics[width=\textwidth]{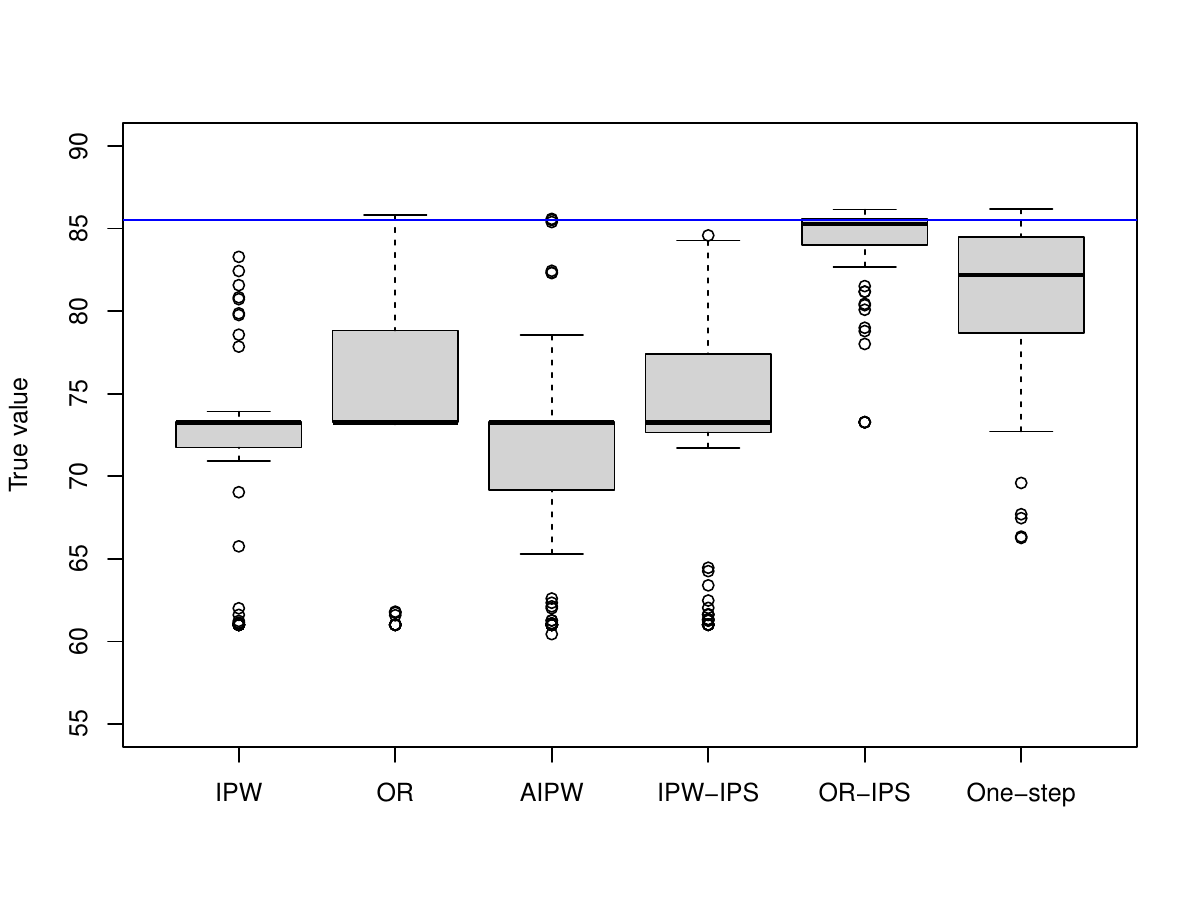}
         \caption{Simulations. }
         \label{fig:fair-dp}
     \end{subfigure}
     \hfill
     \begin{subfigure}[b]{0.475\textwidth}
         \centering
         \includegraphics[width=\textwidth]{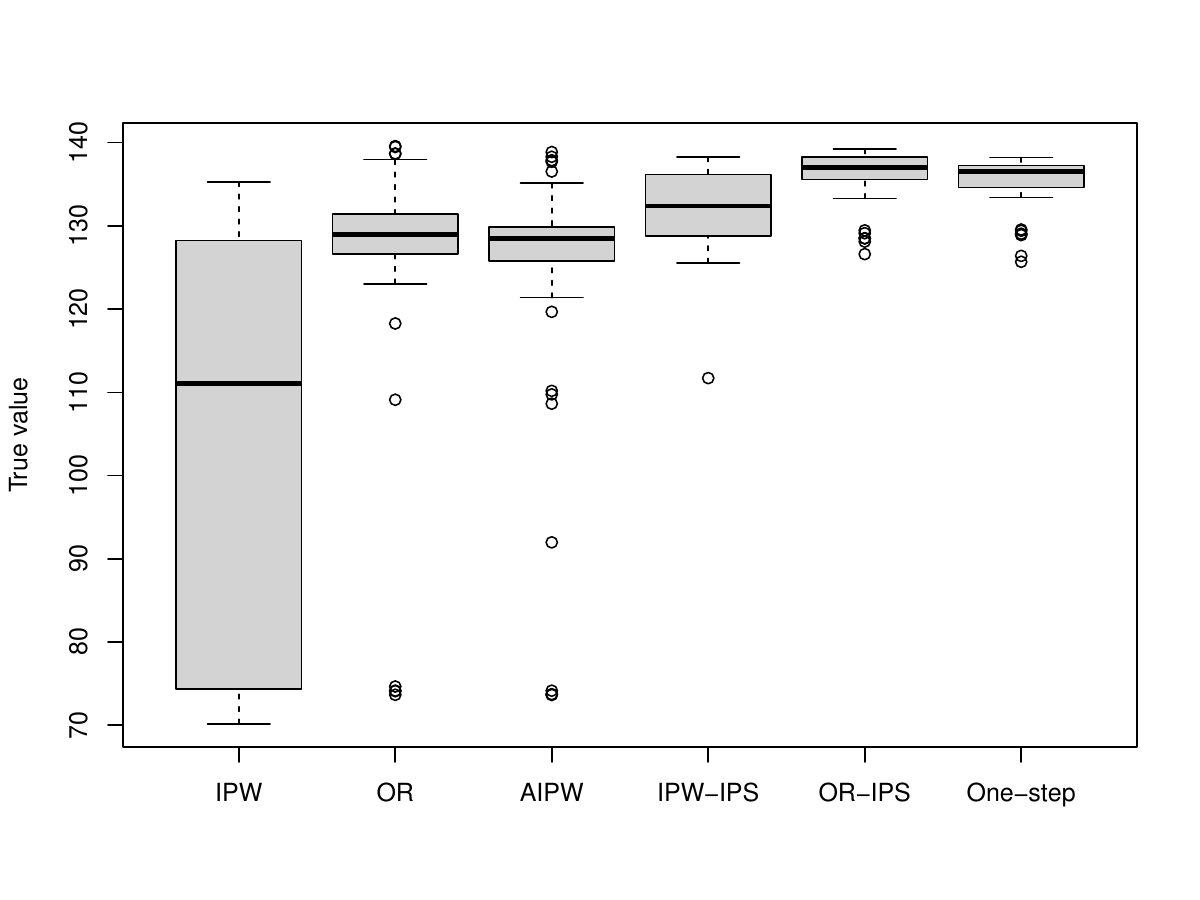}
         \caption{Diabetes data application. }
         \label{fig:diabetes-data}
     \end{subfigure}
        \caption{Performance of optimal policies under three standard methods (IPW, OR, AIPW) and our proposed three methods (IPW-IPS, OR-IPS, One-step). The blue line is the (approximate) true optimal value.}
\end{figure}

Additional simulation results are given in Section~\ref{sm.sec:simu} of the Supplementary Material. Specifically, we illustrate that our proposed policy learning methods have comparable performance when there is no positivity violation, and also illustrate the better performance of our proposed methods when using parametric models.

\subsection{Data application}

We illustrate our proposed methods using semi-synthetic data from the \texttt{Fairlearn} open source project \citep{weerts2023fairlearn}. Additional information on our data analysis is provided in Section~\ref{sm.sec:data} of the Supplementary Material.

The Diabetes dataset represents ten years (1999-2008) of clinical care at 130 US hospitals and integrated delivery networks \citep{strack2014impact}, and contains hospital records of patients diagnosed with diabetes who underwent laboratory tests and medications and stayed up to 14 days. Our application aims to learn the optimal policy for prescribing diabetic medication by maximizing the expected outcome under the demographic parity constraint. The sensitive attribute is race, and a violation of positivity exists in the data.

We include $7$ baseline covariates: \texttt{race}, \texttt{gender}, \texttt{age}, \texttt{time\_in\_hospital} (number of days between admission and discharge), \texttt{num\_lab\_procedures} (number of lab tests performed during the encounter), \texttt{num\_medications} (number of distinct generic names administered during the encounter) and \texttt{number\_diagnoses} (number of diagnoses). Under positivity violation, we are unable to identify the value function, e.g. relying on the outcome regression's extrapolation to learn the counterfactual outcomes on test data. Thus the potential outcomes are simulated as follows: $Y(0) \sim \mathcal{N}\{20 (1 + {\tt gender} - {\tt age} + {\tt time\_in\_hospital} + {\tt num\_lab\_procedures} + {\tt num\_medications} + {\tt num\_medications}^{2} + \exp{({\tt number\_diagnoses})}), 20^2\}$, and $Y(1) \sim \mathcal{N}\{20 (1 + {\tt gender} - {\tt age} + {\tt time\_in\_hospital} + {\tt num\_lab\_procedures} + {\tt num\_medications} + {\tt num\_medications}^{2} + \exp{({\tt number\_diagnoses})}) + 25 (3 - 5 {\tt age} + 2 {\tt time\_in\_hospital} - 3 {\tt num\_medications} + {\tt race}), 20^2\}$. The estimation setup and policy classes are the same as previous simulations. We run $50$ repetitions; each time we randomly select $500$ patients as training data to learn the optimal policy and $2000$ patients as test data to evaluate the performance. Empirical results are reported in Figure~\ref{fig:diabetes-data}. When the positivity violation is severer, the IPW estimator has extremely large variability, and we also observe that our proposed methods perform consistently better than the standard methods.

\section{Discussion}
\label{sec:disc}

This article proposes a general positivity-free stochastic policy learning framework using observational data, possibly subject to application-specific constraints. There are several interesting directions for future research. It is relevant to extend our methods to the more general case with multiple time points for treatment assignment, multiple treatment levels, or high-dimensional models \citep{WEI2023105460}, where positivity is even more likely to be violated. The incremental propensity score approach can also be extended to account for common issues such as covariate shift \citep{zhao2023efficient,lei2023policy}, censoring and dropout \citep{cui2023estimating}, and truncation by death \citep{chu2023multiply}.

\subsubsection*{Acknowledgements}

Pan Zhao and Julie Josse are supported in part by the French National Research Agency ANR-16-IDEX-0006. Shu Yang is partially supported by NSF SES 2242776, NIH 1R01AG066883 and 1R01ES031651,  and FDA 1U01FD007934.

\bibliographystyle{plainnat}
\bibliography{Bibliography}

\newpage
\appendix

\begin{center}
{\large\bf SUPPLEMENTARY MATERIAL}
\end{center}

\section{Proof of Proposition~\ref{prop:iden}}

The proof of our identification results is straightforward, following similar arguments in \cite{kennedy2019nonparametric}. First, we prove the OR-IPS formula:
\begin{align*}
V(d) &= E[Y(d)] \\
&= E[Y(1)d(X) + Y(0)(1 - d(X))] \\
&= E[E[Y(1)d(X) + Y(0)(1 - d(X)) \mid X]] \\
&= E[d(X)E[Y(1) \mid X] + (1 - d(X))E[Y(0) \mid X]] \\
&= E[d(X)E[Y \mid X, A = 1] + (1 - d(X))E[Y \mid X, A = 0]] \\
&= E\left[\frac{\delta(X) \pi(X)}{\delta(X) \pi(X) + 1 - \pi(X)} \mu_{1}(X) + \frac{1 - \pi(X)}{\delta(X) \pi(X) + 1 - \pi(X)} \mu_{0}(X)\right] \\
&= E\left[\frac{\delta(X) \pi(X) \mu_1 (X) + \{1 - \pi(X)\} \mu_0(X)}{\delta(X) \pi(X) + 1 - \pi(X)}\right].
\end{align*}

Next, we prove the IPW-IPS formula:
\begin{align*}
& E\left[\frac{Y\{\delta(X) A + 1 - A\}}{\delta(X) \pi(X) + 1 - \pi(X)}\right] \\
&= E\left[\frac{Y A \delta(X)}{\delta(X) \pi(X) + 1 - \pi(X)} + \frac{Y (1 - A)}{\delta(X) \pi(X) + 1 - \pi(X)}\right] \\
&= E\left[\frac{Y(1) A \delta(X)}{\delta(X) \pi(X) + 1 - \pi(X)} + \frac{Y(0) (1 - A)}{\delta(X) \pi(X) + 1 - \pi(X)}\right] \\
&= E\left[E\left[\frac{Y(1) A \delta(X)}{\delta(X) \pi(X) + 1 - \pi(X)} + \frac{Y(0) (1 - A)}{\delta(X) \pi(X) + 1 - \pi(X)} \mid X \right]\right] \\
&= E\left[\frac{E[Y(1) A \mid X] \delta(X)}{\delta(X) \pi(X) + 1 - \pi(X)} + \frac{E[Y(0) (1 - A) \mid X]}{\delta(X) \pi(X) + 1 - \pi(X)}\right] \\
&= E\left[Y(1) \frac{E[A \mid X] \delta(X)}{\delta(X) \pi(X) + 1 - \pi(X)} + Y(0) \frac{E[(1 - A) \mid X]}{\delta(X) \pi(X) + 1 - \pi(X)}\right] \\
&= E[E[Y(1)d(X) + Y(0)(1 - d(X)) \mid X]] \\
&= V(d).
\end{align*}

\section{Proof of Proposition~\ref{prop:eif}}
\label{sm.sec:proof.eif}

We derive the efficient influence function for the following statistical functional:
\begin{equation*}
\Psi (P) = E_{P} \left[\frac{\delta(X) \pi(X) \mu_1 (X) + \{1 - \pi(X)\} \mu_0(X)}{\delta(X) \pi(X) + 1 - \pi(X)}\right].
\end{equation*}

For a given distribution $P$ in the nonparametric statistical model $\mathcal{M}$, we let $p$ denote the density of $P$ with respect to some dominating measure $\nu$. For all bounded $h \in L_{2}(P)$, define the parametric submodel $p_{\epsilon} = (1 + \epsilon h) p$, which is valid for small enough $\epsilon$ and has score $h$ at $\epsilon = 0$. We would establish that $\Psi (P)$ is pathwise differentiable with respect to $\mathcal{M}$ at $P$ with efficient influence function $\phi(P)$ if we have that for any $P \in \mathcal{M}$,
\begin{equation*}
\frac{\partial}{\partial\epsilon} \Psi (P_{\epsilon}) \bigg|_{\epsilon = 0} = \int \phi (P) (o) h(o) dP(o).
\end{equation*}

We denote $\pi_{\epsilon} (x) = E_{P_{\epsilon}}[A \mid X = x]$, $\mu_{a,\epsilon} (x) = E_{P_{\epsilon}}[Y \mid X = x, A = a]$, $S = \partial \log p_{\epsilon} / \partial\epsilon$, and can compute
\begin{align*}
\frac{\partial}{\partial\epsilon} \Psi (P_{\epsilon}) \bigg|_{\epsilon = 0} &= \frac{\partial}{\partial\epsilon} E_{P_{\epsilon}}\left[\frac{\delta(X) \pi_{\epsilon}(X) \mu_{1,\epsilon} (X) + \{1 - \pi_{\epsilon}(X)\} \mu_{0,\epsilon}(X)}{\delta(X) \pi_{\epsilon}(X) + 1 - \pi_{\epsilon}(X)}\right] \bigg|_{\epsilon = 0} \\
&= \frac{\partial}{\partial\epsilon} E_{P}\left[(1 + \epsilon S) \frac{\delta(X) \pi_{\epsilon}(X) \mu_{1,\epsilon} (X) + \{1 - \pi_{\epsilon}(X)\} \mu_{0,\epsilon}(X)}{\delta(X) \pi_{\epsilon}(X) + 1 - \pi_{\epsilon}(X)}\right] \bigg|_{\epsilon = 0} \\
&= E_{P}\left[S \frac{\delta(X) \pi(X) \mu_1 (X) + \{1 - \pi(X)\} \mu_0(X)}{\delta(X) \pi(X) + 1 - \pi(X)}\right] \\
&\quad + E_{P}\left[\frac{1}{\delta(X) \pi(X) + 1 - \pi(X)} \left(\pi(X) \frac{\partial}{\partial\epsilon} \mu_{1,\epsilon}(X) \bigg|_{\epsilon = 0} + \mu_1 (X) \frac{\partial}{\partial\epsilon} \pi_{\epsilon} (X) \bigg|_{\epsilon = 0}\right)\right] \\
&\quad + E_{P}\left[\frac{1}{\delta(X) \pi(X) + 1 - \pi(X)} \left(\{1 - \pi(X)\} \frac{\partial}{\partial\epsilon} \mu_{0,\epsilon}(X) \bigg|_{\epsilon = 0} - \mu_0(X) \frac{\partial}{\partial\epsilon} \pi_{\epsilon} (X) \bigg|_{\epsilon = 0}\right)\right] \\
&\quad - E_{P}\left[\frac{\delta(X) \pi(X) \mu_1 (X) + \{1 - \pi(X)\} \mu_0(X)}{\{\delta(X) \pi(X) + 1 - \pi(X)\}^2}\left(\delta(X) \frac{\partial}{\partial\epsilon} \pi_{\epsilon} (X) \bigg|_{\epsilon = 0} - \frac{\partial}{\partial\epsilon} \pi_{\epsilon} (X) \bigg|_{\epsilon = 0}\right)\right].
\end{align*}

Then we need to compute
\begin{align*}
\frac{\partial}{\partial\epsilon} \pi_{\epsilon} (X) \bigg|_{\epsilon = 0} &= \frac{\partial}{\partial\epsilon} \frac{\pi (X) + \epsilon E_{P}[S A \mid X]}{1 + \epsilon E_{P}[S \mid X]} \bigg|_{\epsilon = 0} \\
&= E_{P}[S A \mid X] - \pi (X) E_{P}[S \mid X] \\
&= E_{P}[S (A - \pi (X)) \mid X],
\end{align*}
and for $a = 0, 1$,
\begin{align*}
\frac{\partial}{\partial\epsilon} \mu_{a,\epsilon} (X) \bigg|_{\epsilon = 0} &= \frac{\partial}{\partial\epsilon} \frac{\mu_{a} (X) + \epsilon E_{P}[S Y \mid X, A = a]}{1 + \epsilon E_{P}[S \mid X, A = a]} \bigg|_{\epsilon = 0} \\
&= E_{P}[S Y \mid X, A = a] - \mu_{a} (X) E_{P}[S \mid X, A = a] \\
&= E_{P}[S (Y - \mu_{a} (X)) \mid X, A = a].
\end{align*}

Combining the above derivations, we obtain that 
\begin{align*}
\phi(P) (O) &= \frac{A \delta(X) \{Y - \mu_1(X)\} + (1 - A) \{Y - \mu_0(X)\} + \delta(X) \pi(X) \mu_1 (X) + \{1 - \pi(X)\} \mu_0(X)}{\delta(X) \pi(X) + 1 - \pi(X)} \\ 
&\quad + \frac{\delta(X) \tau(X) \{A - \pi(X)\}}{\{\delta(X) \pi(X) + 1 - \pi(X)\}^2} - \Psi(P),
\end{align*}
which yields the result.

\section{Proof of Theorem~\ref{thm:donsker}}
\label{sm.sec:proof.thm.donsker}

We first outline the inferential strategy from semiparametric theory. Consider a statistical model $\mathcal{M}$ for distributions $\tilde{P}$, with $P$ denoting the true distribution. Under sufficient smoothness conditions, we have the following von Mises expansion for $\Psi (\tilde{P})$:
\begin{equation*}
\Psi (\tilde{P}) = \Psi (P) - \int \phi(\tilde{P})(o) dP(o) + {\rm Rem}(\tilde{P},P),
\end{equation*}
where $\phi (P)$ is the influence function derived in Section~\ref{sm.sec:proof.eif} such that $\int \phi(P)(o) dP(o) = 0$, and ${\rm Rem}(\tilde{P},P) = O(\|\tilde{P} - P\|^2)$ is a second-order reminder term that we will analyze later.

Let $\hat{P}$ be an estimator of $P$, then we obtain the following one-step estimator of $\Psi (P)$:
\begin{equation*}
\hat{\Psi} = \Psi(\hat{P}) + \int \phi(\hat{P})(o) dP_{n}(o),
\end{equation*}
where $P_{n}$ is the empirical distribution. 

Next, we characterize the asymptotic properties of $\hat{\Psi}$. Note that 
\begin{align*}
\hat{\Psi} - \Psi (P) &= \left\{\Psi(\hat{P}) + \int \phi(\hat{P})(o) dP_{n}(o)\right\} - \Psi (P) \\
&= \left\{\Psi(\hat{P}) - \Psi (P)\right\} + \int \phi(\hat{P})(o) dP_{n}(o) \\
&= - \int \phi(\hat{P})(o) dP(o) + {\rm Rem}(\hat{P},P) + \int \phi(\hat{P})(o) dP_{n}(o) \\
&= \int \phi(\hat{P})(o) d\left\{P_{n}(o) - P(o)\right\} + {\rm Rem}(\hat{P},P) \\
&= \int \phi(P)(o) dP_{n}(o) + \int \left\{\phi(\hat{P})(o) - \phi(P)(o)\right\} d\left\{P_{n}(o) - P(o)\right\} + {\rm Rem}(\hat{P},P).
\end{align*}

Therefore, $\sqrt{n} \left\{\hat{\Psi} - \Psi (P)\right\}$ is expressed as the following three terms:
\begin{align*}
\sqrt{n} \left\{\hat{\Psi} - \Psi (P)\right\} &= \sqrt{n} \int \phi(P)(o) dP_{n}(o) \\
&\quad + \sqrt{n} \int \left\{\phi(\hat{P})(o) - \phi(P)(o)\right\} d\left\{P_{n}(o) - P(o)\right\} \\
&\quad + \sqrt{n} {\rm Rem}(\hat{P},P).
\end{align*}

By the central limit theorem, $\sqrt{n} \int \phi(P)(o) dP_{n}(o)$ is asymptotically normal with the asymptotic variance given by $E[\phi^2(P)(O)]$.

We assume that $\phi(P)$ belongs to a Donsker class, so we have that the centered empirical process
\begin{equation*}
\sqrt{n} \int \left\{\phi(\hat{P})(o) - \phi(P)(o)\right\} d\left\{P_{n}(o) - P(o)\right\} = o_p(1).
\end{equation*}

Finally, we characterize the second-order remainder term:
\begin{equation*}
{\rm Rem}(\hat{P},P) = \Psi (\hat{P}) - \Psi (P) + E_{P}[\phi(\hat{P})(O)].
\end{equation*}

We have that
\begin{equation*}
\Psi (P) = E_{P}\left[\frac{\delta(X) \pi(X) \mu_1 (X) + \{1 - \pi(X)\} \mu_0(X)}{\delta(X) \pi(X) + 1 - \pi(X)}\right],
\end{equation*}
and
\begin{align*}
&E_{P}[\phi(\hat{P})(O)] \\
&= E_{P}\left[\frac{A \delta(X) \{Y - \hat{\mu}_1(X)\} + (1 - A) \{Y - \hat{\mu}_0(X)\} + \delta(X) \hat{\pi}(X) \hat{\mu}_1 (X) + \{1 - \hat{\pi}(X)\} \hat{\mu}_0(X)}{\delta(X) \hat{\pi}(X) + 1 - \hat{\pi}(X)} \right.\\ 
&\left.\qquad\quad + \frac{\delta(X) \hat{\tau}(X) \{A - \hat{\pi}(X)\}}{\{\delta(X) \hat{\pi}(X) + 1 - \hat{\pi}(X)\}^2}\right] - \Psi (\hat{P}).
\end{align*}

Combining the derivations above, we have that
\begin{align*}
\left|{\rm Rem}(\hat{P},P)\right| &\leq \hat{C}_1 \|\hat{\mu}_1(X) - \mu_1(X)\|_{L_2} \times \|\hat{\pi}(X) - \pi(X)\|_{L_2} \\
&\quad + \hat{C}_2 \|\hat{\mu}_0(X) - \mu_0(X)\|_{L_2} \times \|\hat{\pi}(X) - \pi(X)\|_{L_2} \\
&\quad + \hat{C}_3 \|\hat{\pi}(X) - \pi(X)\|^{2}_{L_2},
\end{align*}
where $\hat{C}_{1}$, $\hat{C}_{2}$ and $\hat{C}_{3}$ are $O_{p}(1)$. We assume that $\|\hat{\pi}(x) - \pi(x)\|_{L_2} = o_p(n^{-1/4})$, and $\|\hat{\mu}_{a} - \mu_{a}\|_{L_2} = o_p(n^{-1/4})$ for $a = 0, 1$. Therefore, we have that $\sqrt{n} {\rm Rem}(\hat{P},P) = o_p(1)$. That is, we conclude that
\begin{equation*}
\sqrt{n} \left\{\hat{\Psi} - \Psi (P)\right\} \to \mathcal{N}(0, E[\phi^2(P)(O)]),
\end{equation*}
which completes the proof.

\section{Proof of Theorem~\ref{thm:cross.fit}}

Essentially, we need to prove that the centered empirical process is $o_p(1)$, when we avoid Donsker conditions by using the cross-fitting technique. We first review a useful lemma from \cite{kennedy2020sharp}.

\begin{lemma}\label{lem:cf}
Consider two independent samples $\mathcal{O}_1 = (O_1, \ldots, O_n)$ and $\mathcal{O}_2 = (O_{n+1}, \ldots, O_{N})$ drawn from the distribution $\mathbb{P}$. Let $\hat{f}(o)$ be a function estimated from $\mathcal{O}_2$, and $\mathbb{P}_n$ the empirical measure over $\mathcal{O}_1$, then we have
\begin{equation*}
    (\mathbb{P}_n - \mathbb{P})(\hat{f} - f) = O_{\mathbb{P}}\left(\frac{\|\hat{f} - f\|}{\sqrt{n}}\right).
\end{equation*}
\end{lemma}
\begin{proof}
First note that by conditioning on $\mathcal{O}_2$, we obtain that
\begin{equation*}
    \mathbb{E} \left\{\mathbb{P}_n(\hat{f} - f) \,\big|\, \mathcal{O}_2\right\} = \mathbb{E} (\hat{f} - f \,|\, \mathcal{O}_2) = \mathbb{P}(\hat{f} - f),
\end{equation*}
and the conditional variance is
\begin{equation*}
    var\{(\mathbb{P}_n - \mathbb{P})(\hat{f} - f) \,|\, \mathcal{O}_2\} = var\{\mathbb{P}_n (\hat{f} - f) \,|\, \mathcal{O}_2\} = \frac{1}{n}var(\hat{f} - f \,|\, \mathcal{O}_2) \leq \|\hat{f} - f\|^2 / n,
\end{equation*}
therefore by the Chebyshev's inequality we have that
\begin{equation*}
    \mathbb{P}\left\{\frac{|(\mathbb{P}_n - \mathbb{P})(\hat{f} - f)|}{\|\hat{f} - f\|^2 / n} \geq t \right\} = \mathbb{E}\left[\mathbb{P}\left\{\frac{|(\mathbb{P}_n - \mathbb{P})(\hat{f} - f)|}{\|\hat{f} - f\|^2 / n} \geq t \,\bigg|\, \mathcal{O}_2\right\}\right] \leq \frac{1}{t^2},
\end{equation*}
thus for any $\epsilon > 0$ we can pick $t = 1 / \sqrt{\epsilon}$ so that the probability above is no more than $\epsilon$, which yields the result.
\end{proof}

Next, we characterize the asymptotic properties of the cross-fitted estimator $\hat{\Psi}_{\rm CF}$. Following similar steps as Section~\ref{sm.sec:proof.thm.donsker}, we have that
\begin{equation*}
\sqrt{n} \left\{\hat{\Psi}_{\rm CF} - \Psi (P)\right\} = \sqrt{n} \int \phi(P)(o) dP_{n}(o) + \frac{1}{\sqrt{K}} \sum_{k=1}^{K} \sqrt{n_k} (R_{k,1} + R_{k,2}),
\end{equation*}
where $R_{k,1} = \int \left\{\phi(\hat{P}_{-k})(o) - \phi(P)(o)\right\} d\left\{P_{n,k}(o) - P(o)\right\}$, $R_{k,2} = {\rm Rem}(\hat{P}_{-k},P)$.

We note that 
\begin{align*}
R_{k,1} &= \int \left\{\phi(\hat{P}_{-k})(o) - \phi(P)(o)\right\} d\left\{P_{n,k}(o) - P(o)\right\} \\
&= \int \left\{\xi(\hat{P}_{-k})(o) - \xi(P)(o)\right\} d\left\{P_{n,k}(o) - P(o)\right\},
\end{align*}
where $\xi (P)(o) = \phi(P)(o) + \Psi (P)$, and by Lemma~\ref{lem:cf}, we have that 
\begin{equation*}
\sqrt{n_k} R_{k,1} = O_p \left(\|\xi(\hat{P}_{-k}) - \xi(P)\|_{L_2}\right).
\end{equation*}

Note that
\begin{align*}
&\xi(\hat{P}_{-k})(O) - \xi(P)(O) \\
&= \frac{A \delta(X) \{Y - \mu_1(X)\} + (1 - A) \{Y - \mu_0(X)\}}{\delta(X) \pi(X) + 1 - \pi(X)} - \frac{A \delta(X) \{Y - \mu_1(X)\} + (1 - A) \{Y - \mu_0(X)\}}{\delta(X) \pi(X) + 1 - \pi(X)} \\
&\quad + \frac{\delta(X) \pi(X) \mu_1 (X) + \{1 - \pi(X)\} \mu_0(X)}{\delta(X) \pi(X) + 1 - \pi(X)} - \frac{\delta(X) \pi(X) \mu_1 (X) + \{1 - \pi(X)\} \mu_0(X)}{\delta(X) \pi(X) + 1 - \pi(X)} \\
&\quad + \frac{\delta(X) \tau(X) \{A - \pi(X)\}}{\{\delta(X) \pi(X) + 1 - \pi(X)\}^2} - \frac{\delta(X) \tau(X) \{A - \pi(X)\}}{\{\delta(X) \pi(X) + 1 - \pi(X)\}^2},
\end{align*}
and we assume that $|Y|$ and $|\delta (X)|$ are bounded in probability. By the triangle and Cauchy-Schwarz inequalities, we have that
\begin{align*}
\|\xi(\hat{P}_{-k}) - \xi(P)\|_{L_2} &\leq \hat{C}_{1,-k} \|\hat{\mu}_{0,-k}(X) - \mu_0(X)\|_{L_2} + \hat{C}_{2,-k} \|\hat{\mu}_{1,-k}(X) - \mu_1(X)\|_{L_2} \\
&\qquad + \hat{C}_{3,-k} \|\hat{\pi}_{-k}(X) - \pi(X)\|_{L_2}
\end{align*}
where $\hat{C}_{1,-k}$, $\hat{C}_{2,-k}$ and $\hat{C}_{3,-k}$ are $O_p (1)$. We assume that $\|\hat{\pi}(x) - \pi(x)\|_{L_2} = o_p(n^{-1/4})$, and $\|\hat{\mu}_{a} - \mu_{a}\|_{L_2} = o_p(n^{-1/4})$ for $a = 0, 1$. Therefore, we have that $\sqrt{n_k} R_{k,1} = o_p(1)$.

By the same arguments as Section~\ref{sm.sec:proof.thm.donsker}, we have that $\sqrt{n_k} R_{k,2} = o_p(1)$. That is, we conclude that
\begin{equation*}
\sqrt{n} \left\{\hat{\Psi}_{\rm CF} - \Psi (P)\right\} \to \mathcal{N}(0, E[\phi^2(P)(O)]),
\end{equation*}
which completes the proof.

\section{Proof of Theorem~\ref{thm:cvx.policy}}

In this section, we consider a parametric policy class $\mathcal{D}(H)$ indexed by $\eta \in H$. That is, the off-policy learning task is given by the following optimization problem:
\begin{align*}
\eta^{\ast} = \arg\max_{\eta \in H} & \quad V(\eta), \\
\text{ subject to} & \quad c(\eta) \leq 0,
\end{align*}
and the estimated policy is given by
\begin{align*}
\hat{\eta} = \arg\max_{\eta \in H} & \quad \hat{V}(\eta), \\
\text{ subject to} & \quad \hat{c}(\eta) \leq 0.
\end{align*}

We first review a useful lemma from \cite{shapiro1991asymptotic}.
\begin{lemma}\label{lm:appr.cvx.opt}
Let $H$ be a compact subset of $\mathbb{R}^{k}$. Let $C(H)$ denote the set of continuous real-valued functions on $H$, with $\mathcal{L} = C(H) \times \cdots \times C(H)$ the $r$-dimensional Cartesian product. Let $f(\eta) = (f_{0}, \ldots, f_{r}) \in \mathcal{L}$ be a vector of convex functions. Consider the quantity $\eta^\ast$ defined as the solution to the following convex optimization program:
\begin{align*}
\eta^\ast = & \arg\min_{\eta \in H} f_{0}(\eta), \\
\text{subject to } & f_{j}(\eta) \leq 0, j = 1, \ldots, r.
\end{align*}

Assume that Slater's condition holds, so that there is some $\eta \in H$ for which the inequalities are satisfied and non-affine inequalities are strictly satisfied, i.e. $f_{j}(\eta) < 0$ if $f_{j}(\eta)$ is non-affine. Now consider a sequence of approximating programs, for $n = 1, 2, \ldots$:
\begin{align*}
\hat{\eta}_{n} = & \arg\min_{\eta \in H} \hat{f}_{n,0}(\eta), \\
\text{subject to } & \hat{f}_{n,j}(\eta) \leq 0, j = 1, \ldots, r,
\end{align*}
with $\hat{f}_{n} (\eta) = \left(\hat{f}_{n,0},\ldots,\hat{f}_{n,r}\right) \in \mathcal{L}$. Assume that $r(n) \left(\hat{f}_{n} - f\right)$ converges in distribution to a random element $W \in \mathcal{L}$ for some real-valued function $f(\eta)$. Then
\begin{equation*}
r(n) \left(\hat{f}_{n,0}(\eta)(\hat{\eta}_{n}) - f_{0}(\eta^\ast)\right) \to L,
\end{equation*}
for a particular random variable $L$. It follows that $\hat{f}_{n,0}(\eta)(\hat{\eta}_{n}) - f_{0}(\eta^\ast) = O_{p}(1/r(n))$.
\end{lemma}

By Theorem~\ref{thm:donsker} or \ref{thm:cross.fit}, we have that 
\begin{equation*}
\sqrt{n} \left(\hat{V}(\eta) - V(\eta)\right) = \frac{1}{\sqrt{n}} \sum_{i=1}^{n} \phi_{V}(O_{i};\eta) + o_{p}(1),
\end{equation*}
and by condition (ii), we have that 
\begin{equation*}
\sqrt{n} \left(\hat{c}(\eta) - c(\eta)\right) = \frac{1}{\sqrt{n}} \sum_{i=1}^{n} \phi_{c}(O_{i};\eta) + o_{p}(1),
\end{equation*}
where $\phi_{V}$ and $\phi_{c}$ are the influence functions.

By condition (i) and Lemma~\ref{lm:appr.cvx.opt} with $r(n) = \sqrt{n}$, we obtain the conclusion (ii).

To prove conclusion (i), note that 
\begin{equation*}
V(\hat{\eta}) - V(\eta^\ast) = V(\hat{\eta}) - \hat{V}(\hat{\eta}) + \hat{V}(\hat{\eta}) - V(\eta^\ast),
\end{equation*}
where we have that $V(\hat{\eta}) - \hat{V}(\hat{\eta}) = O_p (n^{-1/2})$, and $\hat{V}(\hat{\eta}) - V(\eta^\ast) = O_p (n^{-1/2})$. Hence, we conclude that $V(\hat{\eta}) - V(\eta^\ast) = O_p (n^{-1/2})$, which completes the proof.

\section{Proof of Theorem~\ref{thm:gc.policy}}

In this section, we follow similar techniques in \cite{li2023trustworthy} and consider the off-policy learning task given by the following optimization problem:
\begin{align*}
d^\ast = \arg\max_{d \in \mathcal{D}} & \, V(d) = \arg\max_{d \in \mathcal{D}} E[\xi(P)(O)], \\
\text{ subject to} & \, c(d) = E[\phi_{c}(P)(O)] \leq 0,
\end{align*}
where $\mathcal{D}$ is a Glivenko–Cantelli class, and the estimated optimal policy is given by
\begin{align*}
\hat{d} = \arg\max_{d \in \mathcal{D}} & \, \hat{V}(d) = \arg\max_{d \in \mathcal{D}} \frac{1}{n} \sum_{i=1}^{n} \xi(\hat{P})(O_i)\\
\text{ subject to} & \, \hat{c}(d) = \frac{1}{n} \sum_{i=1}^{n} \phi_{c}(\hat{P})(O_i) \leq 0.
\end{align*}

By condition (iii) of Theorems~\ref{thm:donsker} or condition (ii) of Theorems~\ref{thm:cross.fit}, we have that both $\{\xi (O;d) : d \in \mathcal{D}\}$ and $\{\phi_{c} (O;d) : d \in \mathcal{D}\}$ are GC classes.

To simplify the notation, let we denote $\mathcal{D}_{c} = \left\{d \in \mathcal{D} : c(d) \leq 0\right\}$, and $\mathcal{D}_{n,c} = \left\{d \in \mathcal{D} : \hat{c}(d) \leq 0\right\}$. First we note that the estimation error can be expressed as 
\begin{equation*}
V(d^\ast) - \hat{V}(\hat{d}) = V_{n}^{(1)} + V_{n}^{(2)} + V_{n}^{(3)},
\end{equation*}
where we define
\begin{align*}
V_{n}^{(1)} &= \max_{d \in \mathcal{D}_{c}} E[\xi(P)(O)] - \max_{d \in \mathcal{D}_{c}} P_{n} \xi(P)(O), \\
V_{n}^{(2)} &= \max_{d \in \mathcal{D}_{c}} P_{n} \xi(P)(O) - \max_{d \in \mathcal{D}_{c}} P_{n} \xi(\hat{P})(O), \\
V_{n}^{(3)} &= \max_{d \in \mathcal{D}_{c}} P_{n} \xi(\hat{P})(O) - \max_{d \in \mathcal{D}_{n,c}} P_{n} \xi(\hat{P})(O).
\end{align*}

We analyze the three terms as follows. We have that 
\begin{align*}
V_{n}^{(1)} &= \max_{d \in \mathcal{D}_{c}} E[\xi(P)(O)] - \max_{d \in \mathcal{D}_{c}} P_{n} \xi(P)(O) \\
&\leq \max_{d \in \mathcal{D}_{c}} \left|E[\xi(P)(O)] - P_{n} \xi(P)(O)\right| \\
&= o_{p}(1),
\end{align*}
and similarly we have that
\begin{align*}
V_{n}^{(2)} &= \max_{d \in \mathcal{D}_{c}} P_{n} \xi(P)(O) - \max_{d \in \mathcal{D}_{c}} P_{n} \xi(\hat{P})(O) \\
&\leq \max_{d \in \mathcal{D}_{c}} \left|P_{n} \{\xi(P)(O) - \xi(\hat{P})(O)\}\right| \\
&= o_{p}(1).
\end{align*}

To analyze $V_{n}^{(3)}$, note that for any $d \in \mathcal{D}$, we have that
\begin{align*}
&E[\phi_{c}(P)(O)] - P_{n} \phi_{c}(\hat{P})(O) \\
&= \{E[\phi_{c}(P)(O)] - P_{n} \phi_{c}(P)(O)\} + \{P_{n} \phi_{c}(P)(O) - P_{n} \phi_{c}(\hat{P})(O)\},
\end{align*}
and $E[\phi_{c}(P)(O)] - P_{n} \phi_{c}(P)(O)$ converges to $0$ uniformly as $\{\phi_{c} (O;d) : d \in \mathcal{D}\}$ is a GC class, and $P_{n} \phi_{c}(P)(O) - P_{n} \phi_{c}(\hat{P})(O)$ converges to $0$ uniformly by condition (ii). 

Hence, $\forall \epsilon > 0$, $\exists N_{1} \in \mathbb{N}$, such that for all $n > N_{1}$, $|E[\phi_{c}(P)(O)] - P_{n} \phi_{c}(\hat{P})(O)| < \epsilon$, by which we obtain that, for all $d \in \mathcal{D}_{c}$, i.e., $E[\phi_{c}(P)(O)] \leq c$, we have that $P_{n} \phi_{c}(\hat{P})(O) < c + \epsilon$. Therefore, we have that $\frac{c}{c + \epsilon} d \in \mathcal{D}_{n,c}$. 

As $\xi(\hat{P})(O)$ is uniformly bounded, there exists a constant $L > 0$ such that for any $d_{1}, d_{2}$, we have that
\begin{equation*}
|\xi(\hat{P})(O;d_{1}) - \xi(\hat{P})(O;d_{2})| \leq L \sup_{x \in \mathcal{X}} |d_{1}(x) - d_{2}(x)|.
\end{equation*}

Thus, $\forall \epsilon > 0$, $\exists N_{1} \in \mathbb{N}$, such that for all $n > N_{1}$,
\begin{align*}
V_{n}^{(3)} &= \max_{d \in \mathcal{D}_{c}} P_{n} \xi(\hat{P})(O) - \max_{d \in \mathcal{D}_{n,c}} P_{n} \xi(\hat{P})(O) \\
&\leq \max_{d \in \mathcal{D}_{c}} P_{n} \xi(\hat{P})(O) - \max_{d \in \frac{c}{c + \epsilon}\mathcal{D}_{c}} P_{n} \xi(\hat{P})(O) \\
&\leq \frac{\epsilon}{c + \epsilon} L,
\end{align*}
and similarly, we can obtain that $\exists N_{2} \in \mathbb{N}$, such that for all $n > N_{2}$,
\begin{equation*}
V_{n}^{(3)} \geq -\frac{\epsilon}{c + \epsilon} L,
\end{equation*}
which in combination implies that $V_{n}^{(3)} = o_{p}(1)$.

Next, we prove our result (ii) for the regret. Note that
\begin{equation*}
V(d^\ast) - V(\hat{d}) = \{V(d^\ast) - \hat{V}(d^\ast)\} + \{\hat{V}(d^\ast) - \hat{V}(\hat{d})\} + \{\hat{V}(\hat{d}) - V(\hat{d})\}.
\end{equation*}

We analyze the three terms as follows. By the same argument for proving (i), we have that 
\begin{align*}
V(d^\ast) - \hat{V}(d^\ast) &= E[\xi(P)(O;d^\ast)] - P_{n} \xi(\hat{P})(O;d^\ast) = o_{p}(1), \\
\hat{V}(\hat{d}) - V(\hat{d}) &= P_{n} \xi(\hat{P})(O;\hat{d}) - E[\xi(P)(O;\hat{d})] = o_{p}(1).
\end{align*}

Also by a similar argument, we have that for any $d \in \mathcal{D}$ and $\epsilon > 0$, $\exists N_{2} \in \mathbb{N}$, for all $n > N_{2}$, $\frac{c}{c + \epsilon}d \in \mathcal{D}_{n,c}$, and 
\begin{align*}
\hat{V}(d^\ast) - \hat{V}(\hat{d}) &= \hat{V}(d^\ast) - \hat{V}\left(\frac{c}{c + \epsilon}d^\ast\right) + \hat{V}\left(\frac{c}{c + \epsilon}d^\ast\right) - \hat{V}(\hat{d}) \\
&\leq \frac{\epsilon}{c + \epsilon} L,
\end{align*}
and also that for any $d \in \mathcal{D}$ and $\epsilon > 0$, $\exists N_{3} \in \mathbb{N}$, for all $n > N_{3}$, $\frac{c}{c + \epsilon} \hat{d} \in \mathcal{D}_{n,c}$, and 
\begin{equation*}
V(d^\ast) - V(\hat{d}) \geq V\left(\frac{c}{c + \epsilon}\hat{d}\right) - V(\hat{d}) \geq -\frac{\epsilon}{c} L,
\end{equation*}
so we conclude that $V(d^\ast) - V(\hat{d}) = o_{p}(1)$, which completes the proof.

\section{Additional simulations}
\label{sm.sec:simu}

In this section, we present additional simulation results.

\subsection{Incremental propensity score policy learning with sufficent overlap}

We examine the performance of our proposed methods by comparison with standard policy learning methods, when sufficient overlap indeed holds. We consider the following data generating process:
\begin{align*}
    (X_1, X_2) & \sim \text{Uniform}(0, 1), \\
    (X_3, X_4) & \sim \mathcal{N}\left\{\left(\begin{smallmatrix} 0 \\ 0 \end{smallmatrix}\right), \left(\begin{smallmatrix} 1 & 0.3 \\ 0.3 & 1 \end{smallmatrix}\right)\right\}, \\
    A & \sim \text{Bernoulli}(\text{expit}(0.3 - 0.4 X_1 - 0.2 X_2 - 0.3 X_3 + 0.1 X_4)), \\
    Y(0) & \sim \mathcal{N}\{20 (1 + X_1 - X_2 + X_3^2 + \exp{(X_2)}), 20^2\}, \\
    Y(1) & \sim \mathcal{N}\{20 (1 + X_1 - X_2 + X_3^2 + \exp{(X_2)}) + 25 (3 - 5 X_1 + 2 X_2 - 3 X_3 + X_4), 20^2\}.
\end{align*}

We perform the vanilla direct policy search tasks without constraint. Hence, the optimal policy is simply $d^{\ast}(x) = I\{3 - 5 X_1 + 2 X_2 - 3 X_3 + X_4 > 0\}$. For standard methods, we consider the policy class of linear rules $\mathcal{D}_{\rm linear} = \{d(x) = I\{(1, x_1, x_2, x_3, x_4) \beta > 0\}: \beta \in \mathbb{R}^{5}, \|\beta\|_{2} = 1\}$. For the incremental propensity score policies, we consider the class $\mathcal{D}_{\rm IPS} = \{d(x) = \delta(x;\beta) \pi(x) / \{\delta(x;\beta) \pi(x) + 1 - \pi(x)\}: \beta \in \mathbb{R}^{5}\}$, which is indexed by $\delta(x;\beta) = \exp{\{(1, x_1, x_2, x_3, x_4) \beta\}}$.

We estimate the outcome regression model $\mu(x)$ and the propensity score $\pi(x)$ using the generalized random forests \citep{athey2019generalized} implemented in the \texttt{R} package \texttt{grf}. The unconstrained optimization problems are solved by the genetic algorithm \citep{sekhon1998genetic} implemented in the \texttt{R} package \texttt{rgenoud}. The sample size is $n = 2000$. We compare the true values of the estimated optimal policies using test data with sample size $N = 10^5$. The true optimal value is approximated using the test data. Simulation results of $100$ Monte Carlo repetition are reported in Figure~\ref{fig:suff-ol}. 

\begin{figure}[ht]
     \centering
     \begin{subfigure}[b]{0.475\textwidth}
         \centering
         \includegraphics[width=\textwidth]{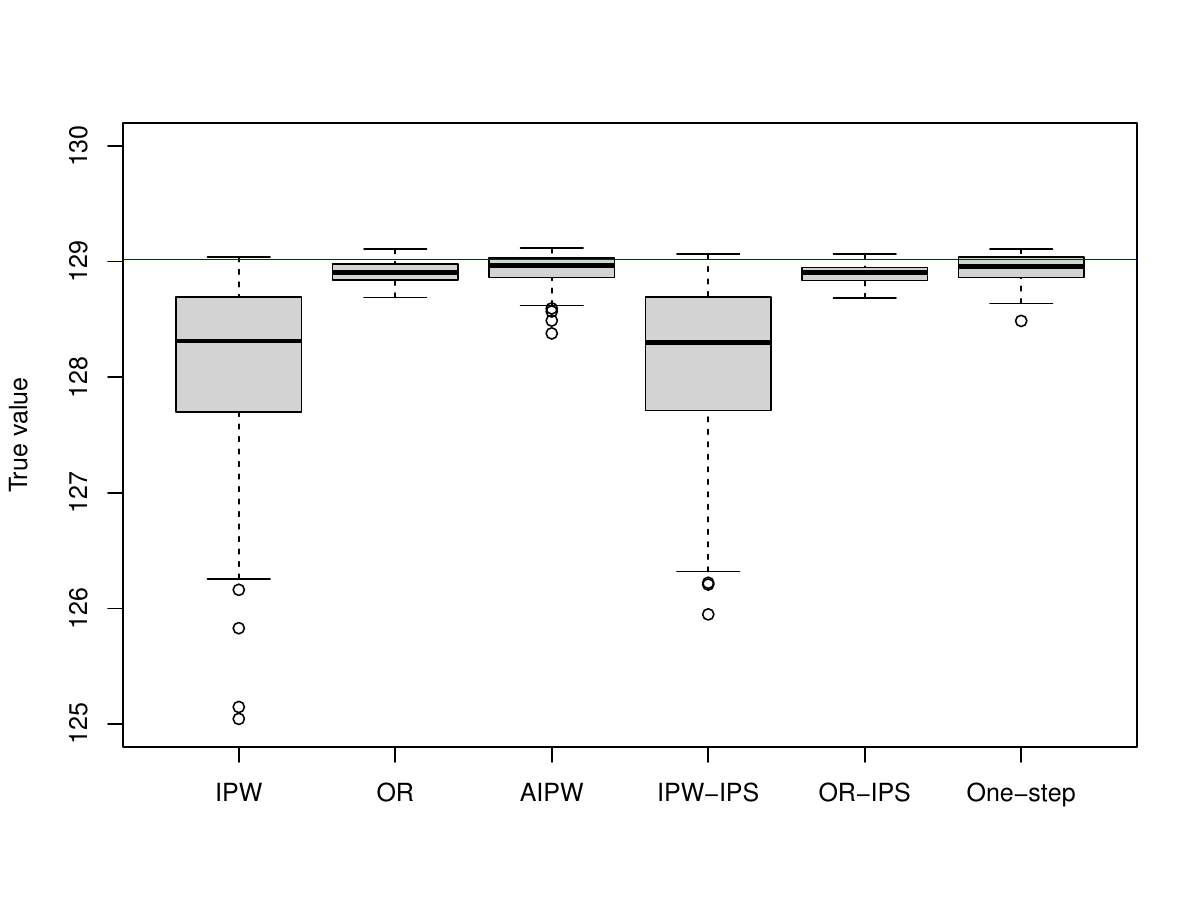}
         \caption{Sufficient overlap. }
         \label{fig:suff-ol}
     \end{subfigure}
     \hfill
     \begin{subfigure}[b]{0.475\textwidth}
         \centering
         \includegraphics[width=\textwidth]{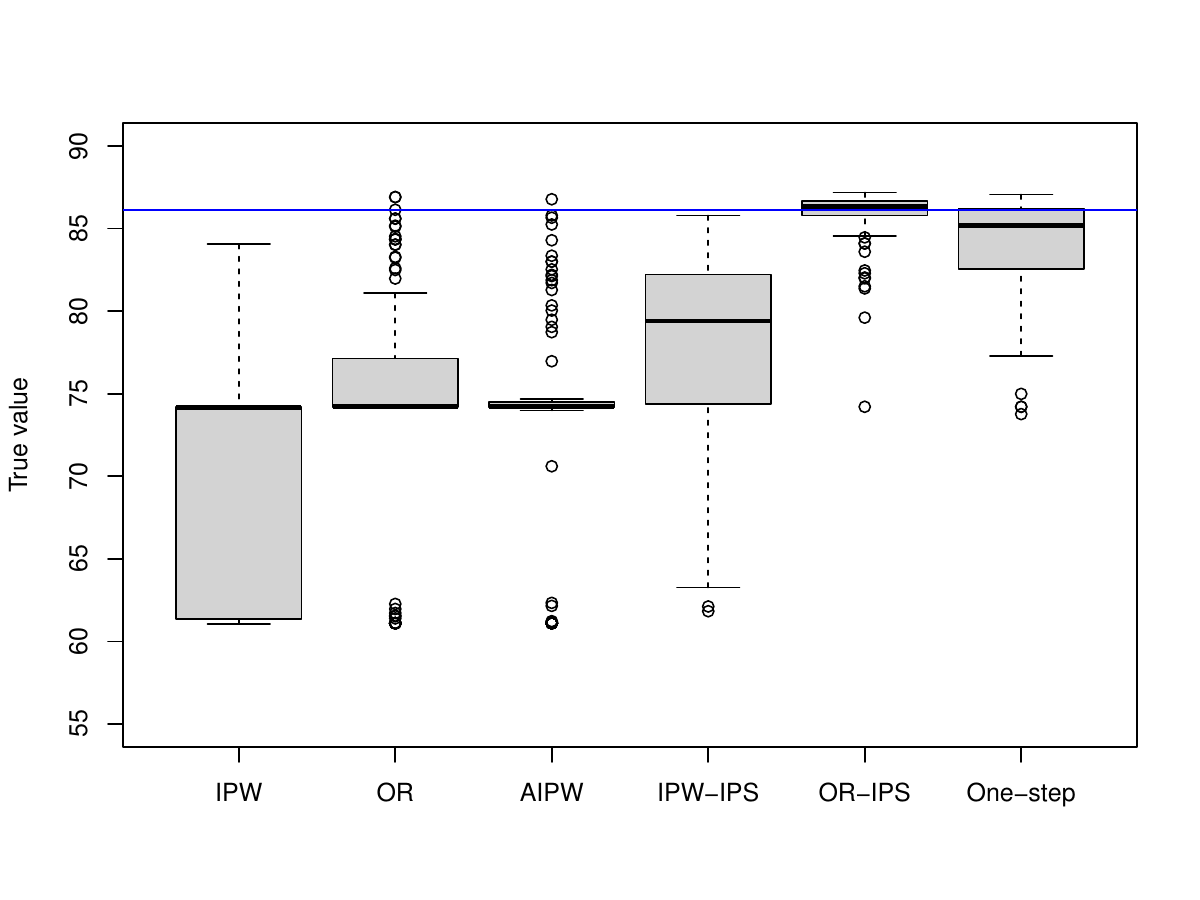}
         \caption{Parametric models. }
         \label{fig:fair-dp-para}
     \end{subfigure}
        \caption{Performance of optimal policies under three standard methods (IPW, OR, AIPW) and our proposed three methods (IPW-IPS, OR-IPS, One-step). The blue line is the (approximate) true optimal value.}
\end{figure}

Despite the fact that the true optimal rule is included in the standard policy class of linear rules but not in our proposed class of incremental propensity score policies, we still observe comparable performance of both classes, which exemplifies the effectiveness of our proposed methods.

\subsection{Incremental propensity score policy learning with parametric models}

We examine the performance of our proposed methods by comparison with standard policy learning methods, when using correctly specified parametric models.

The simulation setup is the same as in the main paper where the positivity assumption is violated, except that the sample size $n = 500$ is smaller and the outcome regression $\mu(s, x)$ and the propensity score $\pi(s, x)$ models are estimated by correctly specified parametric models. Simulation results of $100$ Monte Carlo repetition are reported in Figure~\ref{fig:fair-dp-para}. The standard methods IPW, OR, and AIPW have the worst performance. The IPW-IPS estimator still has large variability, and the OR-IPS and efficient one-step estimators achieve the best performance with the highest value.

\section{Diabetes data analysis}
\label{sm.sec:data}

In this section, we provide supplementary information on our Diabetes data analysis.

The original dataset is available in the UCI Repository \href{https://archive.ics.uci.edu/ml/datasets/Diabetes+130-US+hospitals+for+years+1999-2008}{Diabetes 130-US hospitals for years 1999-2008} \citep{strack2014impact}. The \texttt{Fairlearn} open source project \citep{weerts2023fairlearn} provides full dataset pre-processing script in \texttt{python} on \href{https://github.com/fairlearn/talks/blob/main/2021_scipy_tutorial/preprocess.py}{GitHub}. We follow these pre-processing steps, and provide the \texttt{R} script.

The dataset contains $101766$ patients, and a detailed description of the $25$ variables are available at the \href{https://fairlearn.org/v0.8/user_guide/datasets/diabetes_hospital_data.html}{\texttt{Fairlearn} project}. Originally, the categories of race include ``African American", ``Asian", ``Caucasian", ``Hispanic", ``Other", ``Unknown", and the categories of age include ``$30$ years or younger", ``$30-60$ years", ``Over $60$ years". We dichotomize them, so the resultant categories of race include ``Caucasian" or ``Non-Caucasian", and the resultant categories of age include ``$30$ years or younger" or ``Over $30$ years".

The missing data are completed by multivariate imputation by chained equations, implemented in the \texttt{R} package \texttt{mice}.

\end{document}